\documentclass[conference]{IEEEtran}
\ifCLASSINFOpdf
\else
\fi
\usepackage{cite}
\usepackage{graphicx}
\usepackage{algorithm,algorithmic}
\usepackage[caption=false, font=footnotesize]{subfig}
\usepackage{url}
\usepackage{mathrsfs}
\usepackage{amsfonts}
\usepackage{amsmath}

\DeclareMathOperator*{\argmin}{arg\,min}
\usepackage{amssymb}
\usepackage{mathtools}
\usepackage{algorithm,algorithmic}
\usepackage{amsthm}
\newtheorem{theorem}{Theorem}

\newtheorem{proposition}{Proposition}
\usepackage{bm}
\renewcommand{\vec}[1]{\boldsymbol{#1}}
\usepackage{epstopdf}

\begin{document}
%

\title{Device Scheduling with Fast Convergence for Wireless Federated Learning }

\author{\IEEEauthorblockN{Wenqi Shi, Sheng Zhou, Zhisheng Niu}
        \IEEEauthorblockA{Beijing National Research Center for Information Science and Technology,\\Department of Electronic Engineering, Tsinghua University, Beijing 100084, China\\
        	Email: swq17@mails.tsinghua.edu.cn, \{sheng.zhou,niuzhs\}@tsinghua.edu.cn}
		}
\maketitle

\begin{abstract}
Owing to the increasing need for massive data analysis and model training at the network edge, as well as the rising concerns about the data privacy, a new distributed training framework called federated learning (FL) has emerged.
In each iteration of FL (called round), the edge devices update local models based on their own data and contribute to the global training by uploading the model updates via wireless channels.
Due to the limited spectrum resources, only a portion of the devices can be scheduled in each round.
While most of the existing work on scheduling focuses on the convergence of FL w.r.t. rounds, the convergence performance under a total training time budget is not yet explored.
In this paper, a joint bandwidth allocation and scheduling problem is formulated to capture the long-term convergence performance of FL, and is solved by being decoupled into two sub-problems.
For the bandwidth allocation sub-problem, the derived optimal solution suggests to allocate more bandwidth to the devices with worse channel conditions or weaker computation capabilities.
For the device scheduling sub-problem, by revealing the trade-off between the number of rounds required to attain a certain model accuracy and the latency per round, a greedy policy is inspired, that continuously selects the device that consumes the least time in model updating until achieving a good trade-off between the learning efficiency and latency per round.
The experiments show that the proposed policy outperforms other state-of-the-art scheduling policies, with the best achievable model accuracy under training time budgets.

\end{abstract}


%
\IEEEpeerreviewmaketitle

\section{Introduction}
According to the estimation by Cisco, nearly 850 zettabytes of data will be generated each year at the network edge by 2021 \cite{cisco}. Coupled with the rise of Deep learning \cite{goodfellow2016deep}, these valuable data can bring diverse artificial intelligence (AI) services to the end users. However, training AI models (typically deep neural networks) by  conventional centralized methods is impractical in many scenarios because: 1) uploading raw data to the centralized server via wireless channels will introduce a prohibitively high transmission cost; 2) the devices may hold private data, such as the phone call records, health conditions and location information, and uploading these data to the cloud will cause privacy issues\cite{chen2012data}.

Therefore, a new distributed model training framework called Federated Learning (FL) has emerged \cite{mcmahan2017communication}. In a typical FL system operated in wireless networks, the participating devices are coordinated by a base station (BS) and iteratively perform the local model updating and global aggregation. By updating the model parameters at the end devices, FL can leverage both the data and computation capabilities distributed in the network so as to get a better model with lower latency and preserve the data privacy. As a result, FL becomes a promising technology for distributed data analysis and model training in mobile networks \cite{park2018wireless}, and has been used in many applications such as resource allocation in vehicle-to-vehicle (V2V) communications \cite{samarakoon2018federated} and content suggestions for smartphones \cite{bonawitz2019towards}

However, some key challenges need to be tackled to implement FL in real wireless networks.
Due to the limited wireless resources and stringent delay requirement of FL, only a portion of devices are allowed to upload local models in each round.
Thus the device scheduling policy is critical to FL, and will affect the convergence performance in the following two aspects.
On one hand, the BS needs to wait until all scheduled devices have finished updating and uploading their local models in each round.
Therefore straggler devices with limited computation capabilities or bad channel conditions will significantly slow down the model aggregation.
As a result, scheduling more devices per round will not only reduce the bandwidth allocated to each device and also has higher probability of having a straggler.
On the other hand, scheduling more devices per round increases the convergence rate (w.r.t. the number of rounds) \cite{stich2019local, li2019convergence}, and potentially reduces the number of rounds required to attain the same accuracy.
To this end, the scheduling policy should carefully balance the latency and learning efficiency per round.
Moreover, the scheduling policy should also adapt to the highly dynamic and fluctuating wireless environment.

Recently, implementing FL in wireless networks has received many research efforts.
A branch of research focuses on the novel analog aggregation technologies, in which
the devices upload the updated models simultaneously over a multi-access wireless channel with analog modulation\cite{zhu2018low,amiri2019machine}.
Although the uploading latency can be greatly reduced, very stringent synchronization is required.
Along another series of work, the device scheduling problem is studied.
The convergence performance of FL w.r.t. rounds under three basic scheduling polices, namely random scheduling, round-robin and proportional fair, is analyzed in \cite{yang2019scheduling},
and an energy efficient joint bandwidth allocation and scheduling policy is proposed in \cite{zeng2019energy}.
However, the convergence rate w.r.t. time, which is critical for real-world FL applications, has not been addressed.
To accelerate the FL training, authors of \cite{nishio2019client} propose to  schedule the maximum number of devices in a given time budget for each round while discard the stragglers.
Nevertheless, the time budget is chosen through experiments and can hardly be adjusted under highly-dynamic FL systems.

In this paper, we aim to maximize the convergence
rate of the FL training w.r.t. time rather than rounds. Specifically, we formulate a joint bandwidth allocation and scheduling problem to minimize the expected time for FL training to attain certain model accuracy, and the problem is solved by decoupling into bandwidth allocation and device scheduling sub-problems.
For the bandwidth allocation problem, assuming a given set of scheduled devices, the implicit optimal policy that minimizes the time needed for the current round is derived. We further design an efficient binary search algorithm to get the numerical solution. For device scheduling, a greedy policy is proposed to schedule as many devices as possible to achieve the best trade-off between the latency and learning efficiency per round. Experiments show that the proposed policy outperforms other state-of-the-art scheduling policies in terms of the highest achievable model accuracy under a given training time budget.

\section{System Model}
We consider an FL system consisting one BS and $M$ end devices, denoted by $\mathcal{M}=\{1,2,\dots,M\}$. Each device $i$ has a local data set $\mathcal{D}_i=\{\vec{x}_{i,n}\in \mathbb{R}^s, y_{i,n} \in\mathbb{R}\}_{n=1}^{D_i}$, with $D_i=|\mathcal{D}_i|$ data samples.
Here $\vec{x}_{i,n}$ is the $n^{\text{th}}$ $s$ dimensional input data vector at device $i$, and $y_{i,n}$ is the labeled output of $\vec{x}_{i,n}$.
The whole data set is denoted by $\mathcal{D} = \mathop{\cup} \limits_{i \in \mathcal{M}}\mathcal{D}_i$ with total number of samples $D = \sum \limits_{i \in \mathcal{M}}D_i$, and we assume that all local data sets are non-overlapping with each other.

The goal of the training process is to find the model parameter $\vec{w}$, so as to minimize a particular loss function on the whole data set. The  problem can be expressed as
\begin{equation}
    \min \limits_{\vec{w}} J(\vec{w}) \triangleq
    \min \limits_{\vec{w}} \left\{\frac{1}{D}\sum_{i\in \mathcal{D}} f_i(\vec{w})\right\}, \label{GL}
\end{equation}
where the local loss function $f_i(\vec{w})$ on the data set $\mathcal{D}_i$ is defined as
$f_i(\vec{w}) \triangleq \frac{1}{D_i}\sum_{n \in \mathcal{D}_i} f(\vec{w}, \vec{x}_{i,n}, y_{i,n})$,
and $f(\vec{w}, \vec{x}_{i,n}, y_{i,n})$ captures the error of the model parameter $\vec{w}$ on the training input-output pair $\{\vec{x}_{i,n}, y_{i,n}\}$.

\subsection{Federated Learning over Wireless Network}
FL uses an iterative approach to solve problem \eqref{GL},
and each round, indexed by $k$, contains the following 3 steps.
\begin{enumerate}
    \item The BS broadcasts the current global model $\vec{w}^k$ to all scheduled devices (denoted by $\Pi^k  \subset \mathcal{M}$) that participate the current round.
    \item Each device $i\in \Pi^k$ updates its local model by applying the gradient decent algorithm on its local data set (i.e., $\vec{w}_i^{k+1}=\vec{w}^k-\eta \nabla f_i(\vec{w}^k)$), and uploads the updated model $\vec{w}_i^{k+1}$ to the BS.
    \item After receiving all the uploaded models, the BS aggregates them (i.e., averages the uploaded local models) to generate a new global model
    $\vec{w}^{k+1} = \frac{1}{|\Pi^k|}\sum_{i\in\Pi^k}\vec{w}_i^{k+1}$.
\end{enumerate}

It has been shown in \cite{stich2019local} that for training on the local data sets $\mathcal{D}_i$ that follow i.i.d. distribution, increasing the participating devices can linearly speedup the convergence rate. While for more practical scenarios in FL, the local data sets can be non-i.i.d. \cite{zhao2018federated}. The relation between the number of participating devices and the convergence rate becomes non-linear. As shown in \cite{li2019convergence}, the number of  rounds required to attain a certain model accuracy is on the order of $\mathcal{O}((1+\frac{1}{|\Pi^k|})G + \Gamma)$, where $G$ and $\Gamma$ are parameters that relate to the FL configurations and data distribution. Since different FL applications have various data distribution characteristics, we use the following expression to approximate the number of rounds required to attain certain accuracy, that can adapt to both i.i.d. and non-i.i.d. data distributions
\begin{equation}
N(\Pi^k) = \beta\left(\theta + \frac{1}{|\Pi^k|}\right),
\label{app}
\end{equation}
where the parameters $\theta$ and $\beta$ can be determined through experiments.

\subsection{Latency Model}
We consider an arbitrary round $k$ in the rest of this paper and omit the index $k$ without loss of generality.
\subsubsection{Computation Latency}

To characterize the randomness of the computation latency of local model updating, we use the shifted exponential distribution\cite{lee2017speeding}:
\begin{equation}
\mathbb{P}[t_i^{\text{cp}}<t]=
\begin{cases}
1-e^{-\frac{ \mu_i}{D_i}(t-a_i D_i)} & \text{, $t\geq a_i D_i$,} \\
0 & \text{, otherwise,}
\end{cases}
\label{shifted_exponential}
\end{equation}
where $a_i>0$ and $\mu_i>0$ are parameters that indicate the fluctuation and maximum of the computation capabilities, respectively.
Moreover, we ignore the computation latency of the model aggregation at the BS, because of the relatively stronger computation capability of the BS and low complexity of the model aggregation.

\subsubsection{Communication Latency}
Regarding the local model uploading phase of the scheduled devices, we consider an OFDMA system with total bandwidth $B$. The bandwidth allocated to device $i$ is denoted by $\gamma_i B$, where $\gamma_i$ is the allocation ratio that satisfies $\sum_{i=1}^{M} \gamma_i \leq 1$, and $0 \leq \gamma_i \leq 1$. Therefore, the achievable transmission rate (bits/s) can be written as
$r_i  = \gamma_i B \text{log}_2\left(1+\frac{p_i h_i^2}{N_0}\right)$,
where $p_i$ denotes the transmit power density (Watt/Hz) and is assumed the same over the whole bandwidth, and $h_i$ denotes the corresponding channel gain, and $N_0$ is the noise power. Thus the communication latency of device $i$ is
\begin{equation}
  t_i^{\text{cm}} = \frac{S_\text{model}}{r_i},
\end{equation}
where $S_\text{model}$ denotes the size of $\Vec{w}_i$, in bits.
Since the transmit power of the BS is much higher than that of the devices and the whole downlink bandwidth is used by BS to broadcast the model, we ignore the latency of broadcasting the global model.

\subsubsection{Total Latency per Round}
Due to the synchronous model aggregation of FL, the total latency per round $t^{\text{round}}(\Pi)$ is determined by the slowest device among all the scheduled devices, as thus,
\begin{align}
    t^{\text{round}}(\Pi) = \max_{i\in\Pi} \{t_i^{\text{cm}}+t_i^{\text{cp}}\}.
\end{align}

\section{Joint Bandwidth Allocation and Scheduling}
We consider a joint bandwidth allocation and scheduling problem in order to maximize the convergence rate w.r.t. time. Because the model accuracy statistically increases during the training, it is equivalent to minimize the total latency of FL to attain certain model accuracy, which can be denoted by the number of required rounds (i.e., $N(\Pi)$, given by eq.\eqref{app}) times the latency per round (i.e., $t^{\text{round}}(\Pi)$).

\begin{equation*}
\begin{aligned}
    & (\textrm{P1}) \quad \underset{\Pi, \gamma_i}{\text{min}}
    & & \beta\left(\theta+\frac{1}{|\Pi|}\right) \underset{i\in \Pi}{\text{max}}\left\{\frac{S_\text{model}}{\gamma_i B \text{log}_2\left(1+\frac{p_i h_i^2}{N_0}\right)}
    + t_i^\text{cp}\right\} \\
    & \qquad \quad \text{s.t.}
    & & \Pi \subset \mathcal{M}, \\
    & & & \sum_{i=1}^{M} \gamma_i \leq 1, \\
    & & &  0 \leq \gamma_i \leq 1.
    \label{P1}
\end{aligned}
\end{equation*}

The optimization problem (P1) is hard to solve due to the first constraint and the \textit{max} term in the objective function. Nevertheless, we solve the problem (P1) by decoupling it into two sub-problems so as to inspire good heuristic algorithms.

\subsection{Bandwidth Allocation}
First, we consider the bandwidth allocation problem when the scheduling decision $\Pi$ is given. The sub-problem can be written as follows
\begin{equation*}
\begin{aligned}
    & (\textrm{P2}) \quad \underset{\gamma_i}{\text{min}}
    & & \underset{i\in \Pi}{\text{max}}\left\{\frac{S_\text{model}}{\gamma_i B \text{log}_2\left(1+\frac{p_i h_i^2}{N_0}\right)}
    + t_i^\text{cp}\right\} \\
    & \qquad \quad \text{s.t.}
    & & \sum_{i\in\Pi} \gamma_i \leq 1, \\
    & & & 0 \leq \gamma_i \leq 1.
    \label{P2}
\end{aligned}
\end{equation*}

By solving the problem (P2), the BS can derive the optimal bandwidth allocation (i.e., $\gamma_i$) to minimize the latency of the current round, given the scheduling policy.

\begin{theorem}
The optimal solution of (P2) is as follows
\begin{equation}
   \gamma_i = \frac{S_{\rm{model}}}{(t^*(\Pi)-t_i^{\rm{cp}})B {\rm{log}}_2\left(1+\frac{p_i h_i^2}{N_0}\right)},
   \label{eq2}
\end{equation}
where $t^*(\Pi)$ is the objective value of (P2) that satisfies
\begin{equation}
   \sum_{i\in\Pi} \frac{S_{\rm{model}}}{(t^*(\Pi)-t_i^{\rm{cp}})B {\rm{log}}_2\left(1+\frac{p_i h_i^2}{N_0}\right)} = 1.
   \label{t*}
\end{equation}
\end{theorem}

\begin{proof}
If any device has finished the model updating earlier than other devices, we can reassign its bandwidth to other devices. As a result, the round latency can be shortened. Based on this observation, the optimal solution of (P2) can be achieved if and only if all devices finish updating at the same time. Thus the optimal solution and corresponding objective value is given by the following equations
\begin{equation*}
\left\{
\begin{array}{l}
\frac{S_{\rm{model}}}{\gamma_i B {\rm{log}}_2\left(1+\frac{p_i h_i^2}{N_0}\right)} + t_i^{\rm{cp}} = t^*(\Pi) , \forall i \in \Pi, \\
\sum_{i\in\Pi} \gamma_i = 1 ,   \\
0 \leq \gamma_i \leq 1 .
\end{array}
\right.
\label{eq1}
\end{equation*}
\end{proof}
We need to solve a $|\Pi|$-order equation to get the explicit solution of (P2), which is impossible for $|\Pi| \geq 5$. Since a typical FL system will involve tens or even hundreds of devices in one round \cite{bonawitz2019towards},  we propose a binary search algorithm (Algorithm \ref{alg1}) to get the optimal value of (P2). Begin with the target value $T$ that equals to the upper bound of the initial searching region $[T_{\text{low}}, T_{\text{up}}]$,
we iteratively compute the needed bandwidth for the current target value $T$ according to \eqref{eq2} (steps 5, 6) and
halve the searching region according to whether the bandwidth satisfy the bandwidth constraint (steps 9-14). Given the precision requirement of the searching result (i.e., $\epsilon$), the complexity of the algorithm is on the order of $\mathcal{O}\left(|\Pi|\text{log}_2\left(\frac{T_\text{up}}{\epsilon}\right)\right)$.

\begin{algorithm}
    \caption{Binary Search for the Objective Value of (P2)}
    \label{alg1}
    \begin{algorithmic}[1]
    \renewcommand{\algorithmicrequire}{\textbf{Input:}}
    \renewcommand{\algorithmicensure}{\textbf{Output:}}
    \REQUIRE \quad
    \\ $S_\text{model}$, $B$, $p_i$, $h_i^2$, $N_0$, $t_i^\text{cp}$, $\Pi$, $\epsilon$
    \ENSURE  \quad \\ $T$
    \STATE {Give a big enough $T_\text{up}$, and $T_\text{low} = \text{max}\{t_i^\text{cp}\}, \forall i \in \Pi$}
    \STATE {$T = T_\text{up}$, $success = False$}
    \WHILE {NOT $success$}
        \STATE {$s=0$}
        \FOR{$i \in \Pi$}
            \STATE {$\gamma_i=\frac{S_\text{model}}{(T-t_i^\text{cp})\left[B \text{log}_2\left(1+\frac{p_i h_i^2}{N_0}\right)\right]}$}
            \STATE {$s=s+\gamma_i$}
        \ENDFOR
        \IF {$1-\epsilon \leq s \leq 1$}
            \STATE {$success = True$}
        \ELSIF {$s < 1-\epsilon$}
            \STATE {$T = \frac{T+T_\text{low}}{2}$, $T_\text{up} = T$}
        \ELSIF {$1 < s$}
            \STATE {$T = \frac{T+T_\text{up}}{2}$, $T_\text{low} = T$}
        \ENDIF
    \ENDWHILE
    \RETURN {$T$}
    \end{algorithmic}
\end{algorithm}

\subsection{Latency-Learning Efficiency Trade-off}

To capture the trade-off between the latency and learning efficiency per round, we have the following proposition.
\begin{proposition}
If $\Pi$ is randomly sampled from $\mathcal{M}$ and all devices are homogeneous (with local data set size $d$, transmit power $p$ and channel gain $h$), then the expectation of total updating latency per round can be bounded by
\begin{equation}
\begin{split}
    & ad + \frac{d}{|\Pi|\mu} + |\Pi|\mathbb{E}\left\{\frac{S_{\rm{model}}}{B{\rm{log}}_2(1+\frac{p h^2}{N_0})}\right\}
    \leq \mathbb{E}\{t^*(\Pi)\} \\
    \leq & ad+ \frac{|\Pi|d}{\mu}\sum_{i=1}^{|\Pi|}\frac{1}{i} + |\Pi|\mathbb{E}\left\{\frac{S_{\rm{model}}}{B{\rm{log}}_2(1+\frac{p h^2}{N_0})}\right\},
    \label{app_rt}
\end{split}
\end{equation}
where $a,\mu$ are the parameters of the shifted exponential computation latency distribution.
\end{proposition}
\begin{proof}
From \eqref{t*}, we have
\begin{equation}
\begin{split}
& \sum_{i\in\Pi} \frac{S_{\rm{model}}}{(t^*(\Pi)-t_{\rm{min}}^{\rm{cp}})B {\rm{log}}_2\left(1+\frac{p_i h_i^2}{N_0}\right)} \\
\leq
& \sum_{i\in\Pi} \frac{S_{\rm{model}}}{(t^*(\Pi)-t_i^{\rm{cp}})B {\rm{log}}_2\left(1+\frac{p_i h_i^2}{N_0}\right)} = 1,
\end{split}
\end{equation}
where $t_{\rm{min}}^{\rm{cp}}=\min\limits_{i\in\Pi}\{t_i^{\rm{cp}}\}$.
Thus
\begin{equation}
\begin{split}
    \mathbb{E}\{t^*(\Pi)\} &\geq
    \mathbb{E}\left\{t_{\rm{min}}^{\rm{cp}}+\sum_{i\in\Pi} \frac{S_{\rm{model}}}{B {\rm{log}}_2\left(1+\frac{p_i h_i^2}{N_0}\right)}\right\} \\
    &=\mathbb{E}\{t_{\rm{min}}^{\rm{cp}}\} + |\Pi|\mathbb{E}\left\{\frac{S_{\rm{model}}}{B{\rm{log}}_2(1+\frac{p h^2}{N_0})}\right\}.
    \label{lwb}
\end{split}
\end{equation}
By calculating the expectation of the smallest order statistic of $t_i^{\rm{cp}}$, we can get the lower bound of $\mathbb{E}\{t^*(\Pi)\}$.
The upper bound can be derived similarly.
\end{proof}

\begin{figure}[!t]
\centering
\subfloat[]{
\includegraphics[width=0.485\linewidth]{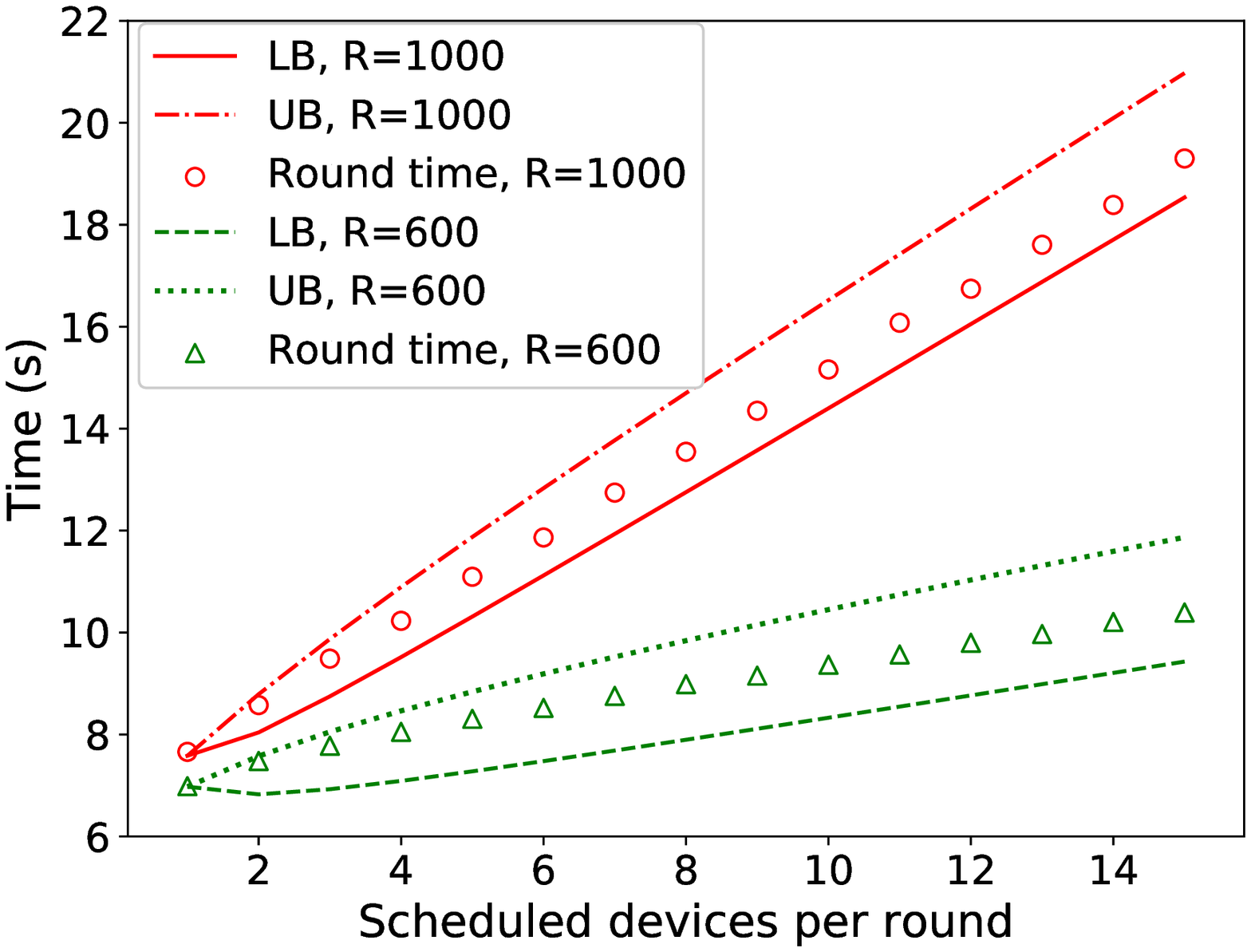}}
\label{random_C_vs_ET}\hfill
\subfloat[]{
\includegraphics[width=0.485\linewidth]{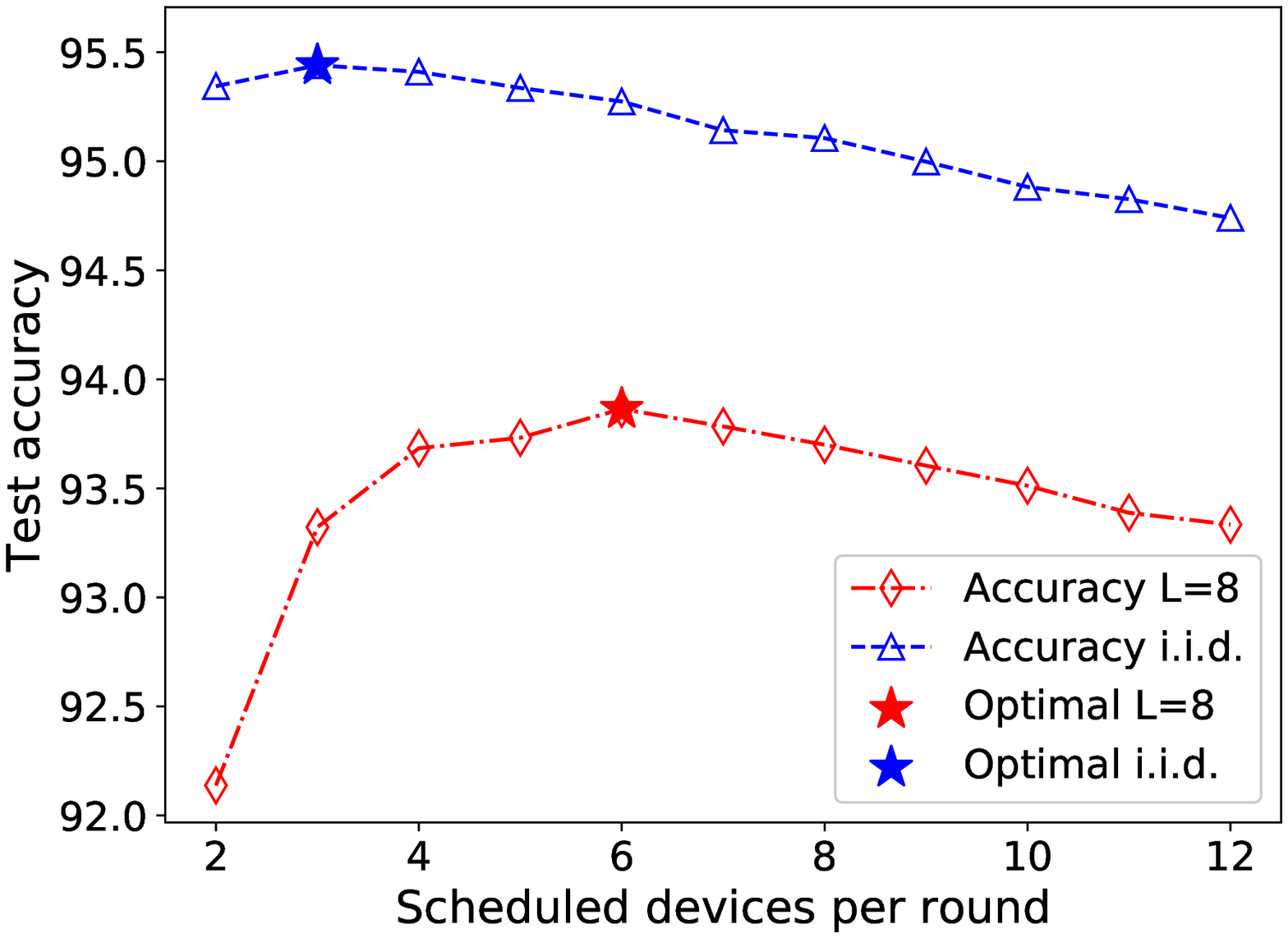}}
\label{random_C_vs_acc}
\caption{(a) The averaged latency per round v.s. the number of devices scheduled with the cell radius $R=600$ and 1000, respectively. The curves denote the bounds in \eqref{app_rt} and the scatters denote the averaged round latency given by Algorithm \ref{alg1} among 1000 trails.  (b) The best accuracy attained under a given training time budget v.s. the number of scheduled devices under different non-i.i.d. level training data sets ($L=8$ and i.i.d.). The time budget is 245 seconds and the results are averaged over 5 trails. For the detail experiment settings, please refer to the next section.}
\label{fig1}
\end{figure}

When $|\Pi|$ is large, the communication latency dominates $t^*(\Pi)$, thus $t^*(\Pi)$ grows almost linearly with $|\Pi|$ and
the lower bound can be used to approximate $t^*(\Pi)$ (as shown in Fig.~\ref{fig1}(a)).
As a result, the objective function in (P1) can be approximated by
\begin{equation}
    \beta\left(\theta+\frac{1}{|\Pi|}\right)\left(ad + \frac{d}{|\Pi|\mu} + |\Pi|\mathbb{E}\left\{\frac{S_{\rm{model}}}{B{\rm{log}}_2(1+\frac{p h^2}{N_0})}\right\}\right),
    \label{obj}
\end{equation}
that leads to a clear trade-off between the latency per round and the number of the rounds, that is: to schedule more devices so as to reach the target accuracy with fewer rounds, or to schedule fewer devices in order to reduce the latency per round.
Further, the trade-off is confirmed by Fig.~\ref{fig1}(b), that shows the highest achievable accuracy versus the number of randomly scheduled devices per round.

\subsection{Device Scheduling}

Given the optimal solution of the bandwidth allocation problem, the scheduling problem (P1) is still a combinatorial optimization problem, which is hard to solve. Therefore we propose a greedy algorithm (Algorithm \ref{alg2}) inspired by the observed latency-learning efficiency trade-off, to optimize device scheduling. In Algorithm \ref{alg2}, the device that consumes the least time in model updating and uploading is iteratively added to the scheduled devices set (step 3), until the objective function of the problem (P1) starts to increase (step 4). The order of Algorithm \ref{alg2} is $\mathcal{O}(|\mathcal{M}|^3)$ (because of calling Algorithm 1 for $\mathcal{O}(|\mathcal{M}|^2)$ times), which is much more efficient than the naive brute force search algorithm on the order of $\mathcal{O}(2^{|\mathcal{M}|})$.
\begin{algorithm}
    \caption{Greedy Scheduling Algorithm}
    \label{alg2}
    \begin{algorithmic}[1]
    \renewcommand{\algorithmicrequire}{\textbf{Input:}}
    \renewcommand{\algorithmicensure}{\textbf{Output:}}
    \REQUIRE \quad
    \\ All available devices in the current round $\mathcal{M} = \{1,2,\cdots,M\}$
    \ENSURE  \quad \\ $\Pi$
    \STATE {$\Pi \gets \emptyset$}
    \WHILE {$|\mathcal{M}|>0$}
        \STATE {$x \gets \argmin \limits _{k\in \mathcal{M}} t^\text{round}(\Pi \cup \{k\}) $}
        \IF {$\beta(\theta+\frac{1}{|\Pi|+1}) t^\text{round}(\Pi \cup \{x\}) > \beta(\theta+\frac{1}{|\Pi|}) t^\text{round}(\Pi)$}
            \STATE {Break}
        \ELSE
            \STATE {$\mathcal{M} \gets \mathcal{M} \setminus \{x\}$}
            \STATE {$\Pi \gets \Pi \cup \{x\}$}
        \ENDIF
    \ENDWHILE
    \RETURN {$\Pi$}
    \end{algorithmic}
\end{algorithm}

\section{Experiment Results}
In this section, we evaluate the training performance of FL under different scheduling policies.

\subsection{Environment and FL Setups}
Unless otherwise specified, we consider $M=20$ devices uniformly located in a cell of radius $R=1$ km, and a BS located at the center of the cell. Assume that all devices will be uniformly re-distributed in the cell at the beginning of each round to reflect mobility. The wireless bandwidth is $B=3$ MHz, and the path loss exponent is $\alpha = 3.76$. The transmit power spectrum density of devices is set to be $p_i = 7$ dBm/MHz, and the power spectrum density of the additive Gaussian noise is $N_0 = -114$ dBm/MHz.

We consider an FL task that classifies handwritten digits using the commonly used MNIST data set \cite{lecun1998gradient}, that has 60,000 training images and 10,000 testing images of the 10 digits. The training data samples are evenly partitioned into all devices (i.e., 3,000 local training images for each device) in the following two ways. For the \textit{i.i.d.} setting, each device is assigned a local training data set that uniformly sampled from all 10 digits. While for the \textit{non-i.i.d.} settings, each device uniformly samples its local data set from different subsets (consist of randomly selected $L$ digits) of the whole training data set, where the parameter $L$ captures the non-i.i.d. level of the local training data sets. We apply a standard multilayer perceptron (MLP) model with one hidden layer of 64 hidden nodes, and use ReLU activation. The MLP model has 50,890 weights, and the model size is about 1.6 megabits when quantized by 32-bits float. For local updating, the mini-batch size is 10 and the learning rate is 0.01. The scheduled devices will perform 1 epoch of local updating, and the computation latency is modeled by the shifted exponential distribution with $a=2$ ms/sample and $\mu=4$ sample/ms.

\subsection{Effect of the Number of Scheduled Devices}
\begin{figure}[!t]
\centering
\includegraphics[width=2.8in]{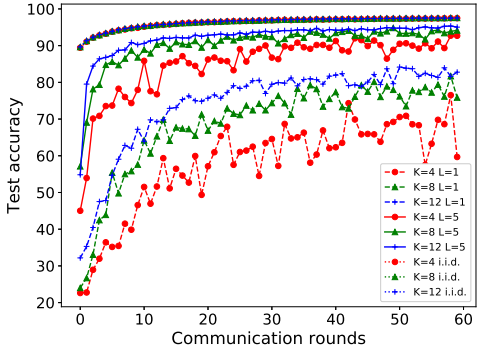}
\caption{The accuracy on the testing data set v.s. the number of rounds. We compare the test accuracy of different number of scheduled devices ($K=4, 8, 12$) under different non-i.i.d. level of the local training data sets ($L=1, 5$ or i.i.d.). Results are averaged over 5 independent trails.}
\label{fg1}
\end{figure}
The effect of the number of scheduled devices on the model accuracy w.r.t. non-i.i.d. level of the local training data sets is shown in Fig.~\ref{fg1}.
In general, the lower non-i.i.d. level of the local training data sets and the more scheduled devices in each round lead to a higher model accuracy within the same number of rounds. Specifically, scheduling 4 devices in each round can achieve only 68.2\% accuracy within 60 rounds on the high non-i.i.d. level data sets ($L=1$). While scheduling the same number of devices under low non-i.i.d. level data sets ($L=5$) and i.i.d. data sets can achieve 90.4\% and 97.4\% accuracy, respectively. Moreover, we notice that as the non-i.i.d. level decreases, the benefit of scheduling more devices will decrease. For highly non-i.i.d. data sets ($L=1$), 9.4\% accuracy improvement (68.2\% to 77.6\%) can be achieved through increasing the number of scheduled devices from $K=4$ to $K=8$, while another 5.2\% (77.6\% to 82.7\%) can be achieved from increasing the number of scheduled devices up to $K=12$. On the contrary, scheduling different number of devices has very little effect on the accuracy for i.i.d. data sets. Then, the regression results of $\beta,\theta$ in \eqref{app} for different non-i.i.d. level data sets are summarized in Tabel.~\ref{tab1} and will be used in the following experiments.

\begin{table}[!t]
\caption{Regression Results of $\beta,\theta$ for Approximating the Number of  Rounds Needed to Attain a Given Accuracy.}
\label{tab1}
\begin{center}
\begin{tabular}{c|c|c|c|c|c}
\hline
\textbf{Data distribution} & $L=1$ & $L=2$ & $L=5$ & $L=8$ & i.i.d.\\
\hline
$\bm{\beta}$  & 123.127 & 103.783 & 89.154 & 63.919 & 27.773 \\
\hline
$\bm{\theta}$  & -0.0448 & -0.0367 & 0.00934 & 0.139 & 0.941 \\
\hline
\end{tabular}
\end{center}
\end{table}

\subsection{Comparison of Different Scheduling Policies}

\begin{figure*}[!t]
  \centering
  \subfloat[]{\includegraphics[width=1.7in]{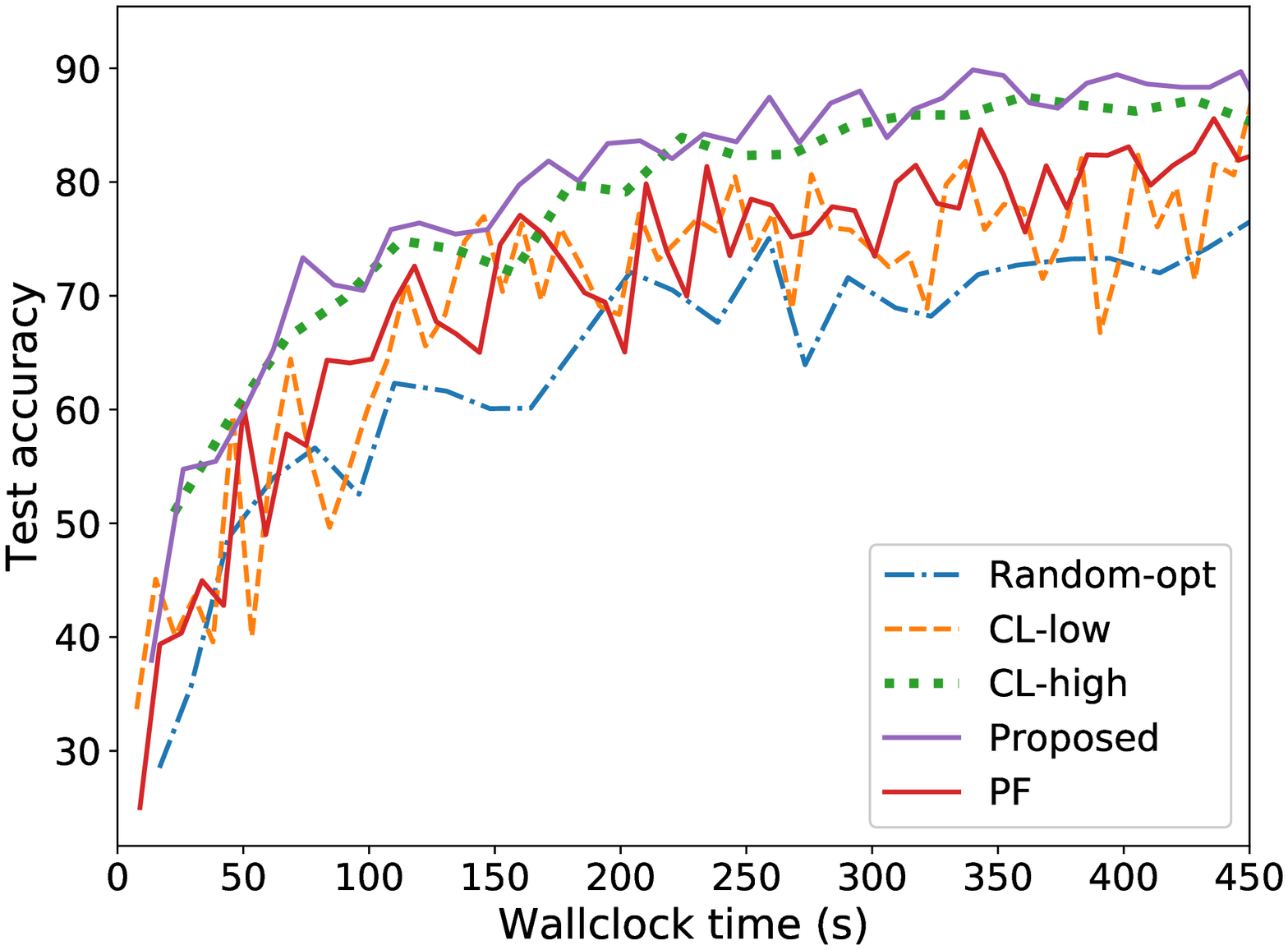}
  \label{comparison_time_vs_acc}}
  \hfil
  \subfloat[]{\includegraphics[width=1.7in]{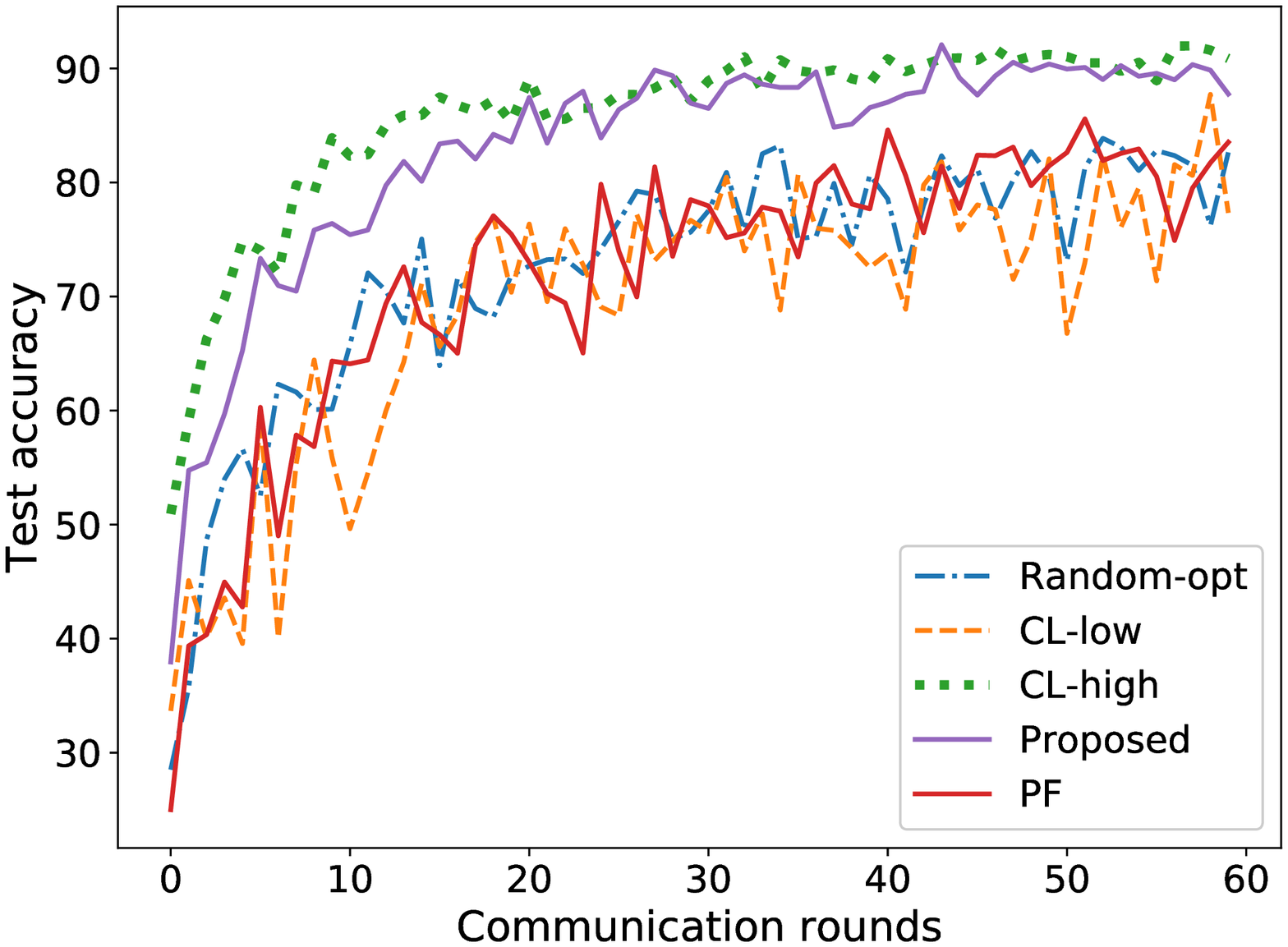}
  \label{comparison_round_vs_acc}}
  \hfil
  \subfloat[]{\includegraphics[width=1.7in]{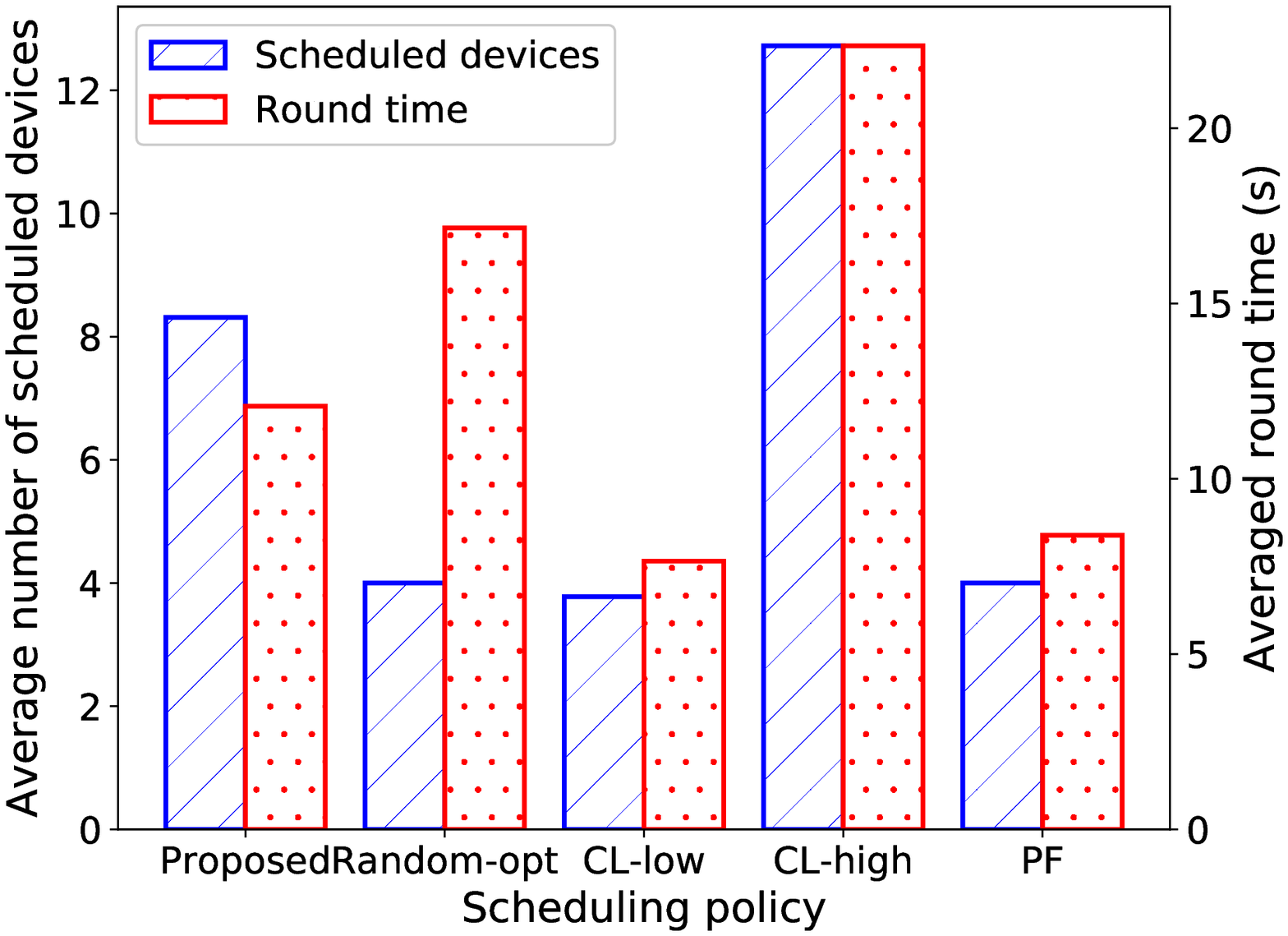}
  \label{chart}}
  \hfil
  \subfloat[]{\includegraphics[width=1.84in]{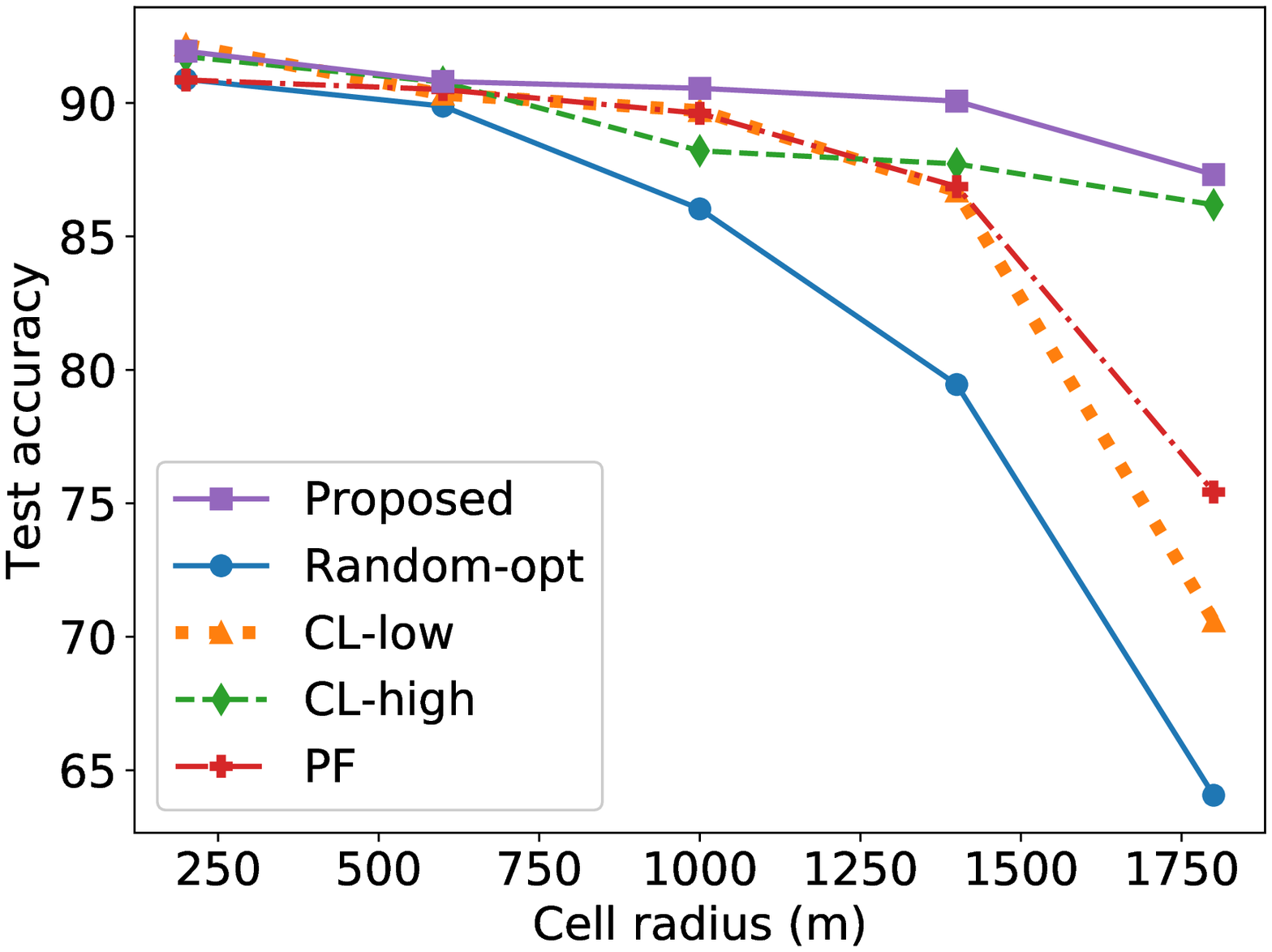}
  \label{different_radius}}
  \caption{ The FL convergence performance under different scheduling polices with the cell radius $R=1400$ meters and non-i.i.d. data set ($L=2$). Results are averaged over 5 trails.
  (a) The test accuracy of the trained model with different scheduling policies v.s. the wallclock time.
  (b) The test accuracy of the trained model with different scheduling policies v.s. the rounds.
  (c) The averaged number of scheduled devices and the corresponding average latency per round w.r.t. different scheduling polices.
  (d) The highest achievable accuracy under a total training time budget equals to 300 seconds v.s. different cell radius.
  }
  \label{fg5}
\end{figure*}

We compare the proposed scheduling policy with 3 baseline scheduling policies. The first baseline policy is the random scheduling policy with the empirically optimal number of scheduled devices $K_\text{opt}$ via experiments. And the second one is the policy proposed in \cite{nishio2019client}, the scheduler iteratively schedules the device that consumes the least time from the set of unscheduled devices until reaching a preset time threshold $T_\text{threshold}$. Here we use two different time thresholds $T_\text{threshold}^\text{low}=8$ seconds and $T_\text{threshold}^\text{high}=25$ seconds, namely CL-low and CL-high, respectively. While the last baseline is the proportional fair policy (PF) proposed in \cite{yang2019scheduling}, that schedules $K$ devices with the best instantaneous channel conditions out of all $M$ devices. In the experiments, we set $K=K_\text{opt}$.

The convergence performances w.r.t. the wallclock time under different scheduling polices are reported in Fig.~\ref{fg5}(a). We can observe from Fig.~\ref{fg5}(a) that the proposed scheduling policy reaches 80\% test accuracy after 171 seconds of FL training, while CL-high needs 224 seconds to attain a similar accuracy and other policies are slower than CL-high. Also note that under a given training time budget (450 seconds), the highest achievable accuracy is 89.85\% with the proposed scheduling policy, that is 14.8\%, 7.47\%, 2.35\% and 3.28\% higher than Random-opt, CL-low, CL-high and PF, respectively. The advantage of the proposed policy is twofold: On one hand, by scheduling the device with better channel condition and lower computation latency, the proposed policy is able to schedule on average 8.31 devices per round within 12.07 seconds as shown in Fig.~\ref{fg5}(c). On the other hand, the proposed policy can achieve a better trade-off between the learning efficiency and latency per round. As shown in Fig.~\ref{fg5}(b) and Fig.~\ref{fg5}(c), although CL-low can reduce the average latency per round to 7.65 seconds, it can only achieve 75\% accuracy within 30 rounds. On the contrary, CL-high converges the fastest w.r.t.  rounds but suffers from longer latency per round (up to 22.72 seconds). As a result, both CL-high and CL-low converge slower than the proposed policy w.r.t. the wallclock time.

The highest achievable accuracy within 300 seconds versus the cell radius $R$ is illustrated in Fig.~\ref{fg5}(d). The results show that the proposed policy can adapt to different cell radii and achieve higher accuracy than all other baseline policies.
It is shown that the advantage of the proposed policy over the baselines will diminish as we reduce the cell radius.
This is because when the cell radius is small, the communication latency will decrease. Therefore all policies tend to schedule more devices, and thus the difference in the convergence rate w.r.t. rounds caused by scheduling different number of devices will also decrease.
We also notice that PF performs similar to Random-opt when the cell radius is small. This is because the computation latency will dominant when the cell radius is small and PF does not take the computation latency into consideration. However, because PF can avoid scheduling the device with bad channel condition, it outperforms Random-opt obviously when the cell radius is large.

\section{Conclusion}
In this paper, we consider a joint bandwidth allocation and scheduling problem to maximize the convergence rate of FL training.
By revealing the trade-off between the learning efficiency and latency per round, a joint bandwidth allocation and scheduling policy is proposed.
The experiments show that the model accuracy under 450 seconds training time budge of the proposed policy is 14.8\%, 7.47\%, 2.35\% and 3.28\% higher than the baseline policies (Random-opt, CL-low, CL-high and PF, respectively).
Further, the proposed policy can adapt to different cell radii and achieve persistently higher accuracy than the baselines.
In the future, the heterogeneity of the system including the different computation capabilities and mobility patterns of the mobile devices can be considered.





%
\bibliographystyle{IEEEtran}
\bibliography{reference}

\end{document}